%% file: main.tex
\def\draft{0}

\documentclass[11pt,letterpaper]{article}
\usepackage[margin=1in]{geometry}
\usepackage[T1]{fontenc}
\usepackage[utf8]{inputenc}
\usepackage{lmodern}
\usepackage{microtype}
\usepackage{mathtools}
\usepackage{amsmath,amssymb,amsthm,amsfonts}
\usepackage{bm}
\usepackage{bbm}
\usepackage{xcolor}
\usepackage{overpic}
\usepackage{physics}
\usepackage[colorlinks = true, 
citecolor = blue,
linkcolor = blue,
urlcolor = blue]{hyperref}
\usepackage{tikz} 
\usetikzlibrary{quantikz2} 
\usepackage[table]{xcolor}
\usepackage{tikz-cd}
\usepackage{ytableau}
\usepackage{graphicx}
\usepackage{xspace}

\input{macro}

\allowdisplaybreaks[1]

\begin{document}

\title{Need for Coherent Access in Constructing
Quantum Cryptography
}

\author{
  Minki Hhan \quad Changhun Oh \quad Vaughn Sohn\\[0.5em]
  {\small KAIST, Daejeon, Korea}\\[0.3em]
  {\small {minkihhan@kaist.ac.kr},
    {changhun0218@gmail.com},
    {webb41@kaist.ac.kr}}
}


\date{\today}

\maketitle

\begin{abstract}    
    We construct quantum oracles relative to which quantum-secure one-way functions (OWFs) exist but pseudorandom states (PRSs) with superlogarithmic output length do not.
    At first glance, this appears to contradict the known black-box constructions of PRS generators from quantum-secure OWFs.
    The distinction lies in the access model to the oracles; our oracle separation uses \emph{classical-accessible} random oracles that can be accessed only classically even by quantum algorithms.
    In fact, our impossibility of PRSs applies to \emph{any} classically accessible classical oracle in place of the random oracle, while keeping the other oracle component unchanged, showing the need for coherent access in constructing PRSs.

    We further show that logarithmic output length pseudorandom function-like states (PRFSs) exist relative to our oracles, giving an oracle separation between classically accessible logarithmic length PRFSs and superlogarithmic length PRSs.
    This shows that fully black-box PRS length extension from logarithmic to superlogarithmic output length must use coherent access to the underlying short PRS.

\end{abstract}
	\clearpage
	\newpage
	\setcounter{tocdepth}{2}
	\tableofcontents
	\newpage

\pagebreak

\section{Introduction}
Finding minimal computational assumptions is a fundamental task in cryptography.
In classical cryptography, one-way functions~(\OWF) are necessary and sufficient for many basic primitives, including pseudorandom generators~(\textsf{PRGs}), pseudorandom functions~(\textsf{PRFs}), and secret-key encryption~\cite{IL89,HILL99,GGM86,Gol04}.

Quantum cryptography has introduced fundamentally new primitives formulated by quantum states and unitaries.
Pseudorandom state~(\PRS) generators are analogous to \textsf{PRGs}: a short classical key specifies an efficiently generated state whose polynomially many copies are computationally indistinguishable from the same number of copies of a Haar-random pure state~\cite{ji_Pseudorandom_2018}.
By default, we consider the length of \PRS to be superlogarithmic length---the logarithmic length will be discussed later.
Similarly, pseudorandom unitaries (\textsf{PRU}) are considered as a quantum analogue of \textsf{PRFs}, which are computationally indistinguishable from Haar random unitaries.
These primitives and their variants are useful for constructing cryptographic schemes such as commitments and encryption~\cite{MY22,AQY22}.

It is well known that both \PRS and \textsf{PRU} can be constructed from quantum-secure \OWF~\cite{ji_Pseudorandom_2018, MH25}. 
At the same time, they are unlikely to imply \OWF because of the oracle separations~\cite{Kre21, Kretschmer_2025}. 
Hence, these quantum primitives are considered weaker cryptographic assumptions than \OWF.

Several constructions of these primitives from \OWF have been suggested for \PRS \cite{BS19,BS20,ABB+24,GB23,JMW23} and {\sf PRU} \cite{MPSY24,FLM+26} together with new constructions and proof techniques.
However, all of these constructions rely heavily on coherent access to the underlying classical primitives.
For example, the standard binary-phase \PRS construction prepares an $m$-qubit state by evaluating a \textsf{PRF} $F_k$ on a superposition of inputs:
\begin{align}
    \ket{\phi_k}
    = 2^{-m/2}\sum_{x\in\{0,1\}^{m}}(-1)^{F_k(x)}\ket{x},
\end{align}
and coherent access is essential to this construction.
Such coherent access requires preserving quantum coherence throughout the evaluation of the underlying classical primitive and may require additional ancilla qubits as well as extra time for uncomputation. 
This could be demanding in practice, particularly on near-term quantum devices, and it is also ironic that realizing potentially weaker primitives may be difficult in practice.

This naturally raises the following question:
\begin{center}
    \emph{Can PRS, or quantum cryptography in general, be constructed from classical primitives without coherent access?}
\end{center}

\subsection{Our results}
We give a negative answer to the above question.
In other words, we show that there exists a quantum channel oracle world in which \OWF exist, but \PRS do not.
This separation result naturally implies that fully black-box constructions of \PRS from \OWF are impossible without coherent access.

\begin{theorem}[Informal version of Theorems~\ref{thm:main-1} and~\ref{thm:main-2}]
    There exists a quantum channel oracle $\mathcal{O}$ relative to which classically computable, quantum-secure \OWF exist, but \PRS with output length $m(\lambda)=\omega(\log\lambda)$ do not exist.
\end{theorem}
The constructions from \OWF to the other primitives in minicrypt \cite{HILL99,GGM86,Nao91} do not involve coherent access; thus the standard classical cryptographic primitives all exist relative to this oracle, including \textsf{PRGs}, \textsf{PRFs} with classical queries, secret-key encryption, message authentication codes, and commitments (Corollary~\ref{cor:classical}).

The main building blocks of our oracle are classically accessible random oracles~\cite{BDF+11,AF22}, to which even quantum algorithms can make only classical queries.
While separations relative to unitary or quantum-accessible classical oracles may better reflect standard-model computations, separations in the classically accessible model are not directly comparable to them in general: coherent access increases the power of both honest algorithms and adversaries~\cite{AF22}.
\footnote{See also \cite{LLPY24} where the classical-accessible oracle was used in the complexity separation.} Furthermore, this shows an exact difficulty of constructing one from another as discussed below.

In fact, we can show the impossibility of \PRS with \emph{any} classical-accessible classical oracles\footnote{To our knowledge, a similar result for any classical oracle was only discussed in \cite{QZ26}.} instead of random oracles while keeping the other components. 
Furthermore, the impossibility holds even when the \PRS generation algorithm is allowed to output slightly mixed states. Allowing mixed-state outputs strengthens our impossibility result and better reflects realistic implementations, where noise or entanglement with discarded ancillas may prevent the output from being exactly pure.


This result shows that coherent access plays a crucial role in constructing \PRS or many quantum cryptographic primitives, especially with large quantum states, from classical ones.
To fully understand the assumptions and scope of the statement that \PRS can be constructed from \OWF, these access requirements must be considered alongside the security properties of the underlying primitives.
From this perspective, our result suggests that the relative strength of minicrypt and microcrypt depends on the oracle access model.
That is, in this model without coherent access, \emph{minicrypt does not suffice for microcrypt!}


\paragraph{\PRS length extension and more.} 
It is now well understood that the length of quantum primitives matters much more than in classical cryptography. While {\sf PRGs} admit a length extension by iterating them, controlling the length of \PRS appears to be much more difficult.\footnote{We note an orthogonal direction to one discussed in this paper: The scalable \PRS and {\sf PRFS} \cite{BS20,BCJZ26} were suggested to decouple the length and security parameter. We remark that the construction of scalable {\sf PRU} is still open \cite{BY26}.} For example, the output length of \PRS is impossible to shrink in a black-box way~\cite{BM24,CGG24}.

The opposite direction of length extension was studied in many works~\cite{LV24,ALY24,BNY25} for special cases, e.g., allowing errors with tomography or quantum sampling. 
It is shown that quantum access to the {\sf PRFS} gives long \PRS, but the classical-accessible case is difficult \cite{BEHMV26}: The impossibility is shown for the long \PRS generation which does not use an ancillary qubit and only uses nonadaptive classical queries to the short {\sf PRFS}.

Our result removes these conditions, showing the following oracle separation.
\begin{corollary}
    With the same oracles, classical-accessible, adaptively secure log-length {\sf PRFS} exists, but \PRS with superlogarithmic output length does not.
\end{corollary}
Again, the \PRS generation algorithm can be noisy as before, use arbitrary polynomially many ancilla qubits, and make adaptive queries to the underlying {\sf PRFS}. 
A similar separation for single-copy-secure \PRS was recently shown in \cite{CCL+26}. 
This shows that the importance of coherent access holds even when the underlying primitives are quantum.

The proof is straightforward, as the output of classical-accessible RO, which can be arbitrarily long polynomially by concatenating them, can be used to generate a single logarithmic-length statistically secure random state \cite{BS20}. The same strategy also shows the impossibility of constructing long \PRS from product single-qubit pseudorandom unitaries \cite{HMX26} as well.

\subsection{Technical overview}

\paragraph{CAROM and $\mathsf{QPSPACE}$ oracle.}
We adopt two oracle components for different purposes.
A random oracle $\mathcal R$ in the classical-accessible random oracle model~(CAROM)~\cite{BDF+11,AF22} provides \OWF, while an independent \textsf{QPSPACE} oracle, following the approach of~\cite{Chen_Power_2025,BMMMY25}, is used to break \PRS.
Here, $\mathcal R$ provides access to a family of random functions $R = \{R_\lambda\}_{\lambda\in\mathbb N}$.
The standard way to give a quantum algorithm access to a classical function is through coherent unitary access.
In contrast, $\mathcal R$ is a quantum channel that first measures the query register in the computational basis and then XORs the function value into the response register.
The \textsf{QPSPACE} oracle supports polynomial-width, oracle-free unitary computations; in particular, its internal computations cannot query $\mathcal R$.
Below, we explain how these two components yield the separation.

\subsubsection{Construction of \OWF}
By the definition of CAROM, each function $R_\lambda$ can be evaluated with a single classical query to $\mathcal R$.
The usual information-theoretic bound for random-function inversion gives negligible success probability for polynomially many classical queries, regardless of the computation performed between queries.
This bound remains valid when the adversary also has access to the independent \textsf{QPSPACE} oracle.
Thus, the family $R$ implements \OWF relative to the full oracle.


\subsubsection{Oracle separation of \PRS from \OWF}
The more delicate part is breaking \PRS.
The original definition of \PRS considers only pure-state outputs~\cite{ji_Pseudorandom_2018}.
We consider a more general setting where the output of a QPT generator may be mixed, but QPT adversaries must still have only negligible advantage in distinguishing copies of its output from copies of a Haar-random pure state. The purity test imposes a near-pure state anyway.

Our starting observation is a more fine-grained version of this observation: 
a simple dichotomy for general mixed states.\footnote{A similar observation for a different setting was recently used in \cite{Hhan26}.}
Consider a pure-state decomposition $\rho=\sum_{i=1}^{N}p_i\ket{\psi_i}\bra{\psi_i}$.
Informally, one of the following holds:
\begin{itemize}
    \item (Case 1) The state $\rho$ has noticeably low purity.
    \item (Case 2) The state $\rho$ is nearly pure and has large overlap with some component $\ket{\psi_i}$ whose weight satisfies $p_i \geq \frac{1}{2N}$.
\end{itemize}
Intuitively, if the state is nearly pure, components with large overlap carry a substantial fraction of the total probability.
Since there are only $N$ components, at least one of them must have sufficiently large weight.
See Lemma~\ref{lem:dichotomy} for the formal statement.

This dichotomy is especially useful when we apply this to the efficient algorithm's output state:
To describe the execution branches, retain discarded registers and defer all measurements to the end, including a final measurement of all nonoutput registers.
Writing $U_A$ for the resulting unitary dilation, the state before measurement must be
\begin{align}
    U_A\ket{0} = \sum_{i} \sqrt{p_i} \ket{\psi_i}\ket{{\sf d}_i},
\end{align}
so that each component state can be written as
\begin{align}
    \sqrt{p_i}\ket{\psi_i}\ket{{\sf d}_i} = (I\otimes \ketbra{{\sf d}_i}) U_A \ket 0.
\end{align}
That is, the original algorithm gives another algorithm that applies a unitary $U_i$ followed by a binary measurement that accepts with probability $p_i$, and in that case outputs the state $\ket{\psi_i}$. In other words, an efficiently generated mixed state admits a decomposition of pure states that can be generated by some efficient algorithms (albeit with small success probability).

Note that $U_i$ specifies all the queries to the oracle, as ${\sf d}_i$ includes all query inputs that must be measured anyway.\footnote{In the actual proof, we do something more complicated because of the adaptivity. We discuss this point later.}
Furthermore, $\log N$ is bounded by the number of binary measurements made by the algorithm, including the oracle measurements and the final measurements of nonoutput registers, and is therefore polynomially bounded.


We apply this observation to the generator's output state.
This allows us to construct a two-stage distinguisher. 
If the candidate \PRS has low purity, repeated SWAP tests detect an antisymmetric outcome with constant probability, whereas copies of a Haar-random pure state never produce this outcome.

In the other case, the dichotomy guarantees a pure branch $|\psi_i\rangle$ with weight $p_i\ge \frac{1}{2N}$ and large overlap with the \PRS. 
This branch provides a candidate for an overlap test. We combine the tests for all candidate branches using the quantum OR lemma~\cite{Harrow_Sequential_2017}.

We use our \textsf{QPSPACE} oracle to run a unitary implementation of the OR test, following~\cite[Appendix~A]{Chen_Power_2025}.
Two obstacles remain before we can apply the OR test: (i) the useful branch may have exponentially small preparation probability, and (ii) its preparation circuit $U_i$ accesses the random oracle $\mathcal R$, which must be removed for the \textsf{QPSPACE} implementation.\footnote{The access to \textsf{QPSPACE} can be dealt with using the \textsf{QPSPACE} oracle, as discussed in~\cite{Chen_Power_2025}.}

\paragraph{How to increase the branch preparation probability.}
To overcome the first obstacle, we amplify the probability of preparing a candidate branch while preserving its normalized output state.
For now, suppose that each candidate execution index $i$ has an associated unitary preparation circuit $U_i$ without accessing $\mathcal R$; we explain how to construct these circuits without $\mathcal R$ below.
The circuit $U_i$ prepares the branch state $\ket{\psi_i}$ with probability $p_i$.
For the branch guaranteed by Case 2, applying $U_i^\dagger$ to the challenge state $\rho$ padded with zero ancillas and accepting when all registers are zero gives 
\begin{align}
    \Pr[\mathrm{accept}]=p_i\langle\psi_i|\rho|\psi_i\rangle \ge \frac{\langle\psi_i|\rho|\psi_i\rangle}{2N}
\end{align}
Even when the overlap is close to one, the factor $p_i$ may be exponentially small, so this test does not meet the completeness requirement of the quantum OR lemma.

We remove this small prefactor using Grover's fixed-point amplitude amplification~\cite{Grover_fixed_2005}.
Although the exponentially small preparation probability requires an exponentially large amplified circuit, we can show that this circuit can be printed in exponential time. 
Therefore, ignoring the random oracle part, the OR test detects the existence of a candidate with sufficiently large overlap. It remains to remove the random oracle from the circuit $U_i$.

\paragraph{How to prepare each branch without $\mathcal R$.}
We now address the second obstacle.
Our idea is simple: we enumerate all possible answers to the random oracle queries, which can be described by classical strings thanks to the classically accessible condition.
It turns out that we need to be careful about the measurements, as different measurement outcomes may induce different classical random oracle queries.

For example, suppose a generator prepares $\frac{|{0^\lambda}\rangle+\ket{k}}{\sqrt{2}}$ for $k\neq 0^\lambda$, measures this register, queries $\mathcal R$ on the outcome, and uses the reply to control a unitary that prepares its pure output.
This situation introduces two different descriptions of the oracle query inputs; if the measurement outcome was $0^\lambda$ and the oracle answer was $y$, the corresponding description may be $(0^\lambda,y)$. Otherwise, if the measurement outcome was $k$ and the oracle answer was $z$, the description becomes $(k,z)$. 
For each branch index $i$, we enumerate \textit{all possible oracle-query records} $d$.
Recall $U_i$ already specifies all the query inputs. Therefore, given the record $d$, 
we consider the unitary $U_i$ with the hardwired oracle answers in $d$, making $U_{i,d}$. 
This gives a family of oracle-free candidate circuits containing, for every actual branch, a circuit that reproduces its output state and preparation probability upon postselection, so that we can apply the strategy above.

We stress that the actual situation is a little more involved.
The main problem is that we first fix $i$, including the oracle query inputs.
These may depend on the oracle replies specified when we choose $d$.
Lemma~\ref{lem:oracle_free} shows that the resulting oracle-free candidate family contains a circuit reproducing each actual branch.





\paragraph{Combining the tests.}
Let $L=L(\lambda)\geq\lambda$ be a polynomial bound on the candidate description length, so there are at most $2^L$ candidates.
We use $L$ copies of the challenge state for the OR test.
For each candidate, we repeat the inverse-circuit test above on $L$ challenge copies, using its amplified preparation circuit,
and accept only if all $L$ tests accept.
Choosing the purity threshold and amplification accuracy appropriately ensures that some candidate test accepts with probability greater than $1/2$ in Case 2.
Together, the purity tests and the OR test therefore accept \PRS with constant probability.

In contrast, for a Haar-random challenge, the Haar-moment identity gives $\mathbb E_{|\phi\rangle\leftarrow\mu_m} |\langle\psi|\phi\rangle|^{2L} = \binom{2^m+L-1}{L}^{-1}$ for every fixed pure state $|\psi\rangle$. Here $\mu_m$ is the Haar measure on $m$-qubit pure states.
Since each candidate test has average acceptance probability at most this quantity, the quantum OR lemma implies that the OR test accepts a Haar-random challenge with negligible probability when $m(\lambda)=\omega(\log\lambda)$.

All candidate preparation circuits, amplified preparation circuits, and the OR test circuit used in this procedure have polynomial width and can be printed in exponential time by a classical Turing machine with a polynomial-length description that never queries $\mathcal R$.
The adversary can therefore execute the OR test with one classical-command query to $\mathsf{QPSPACE}$.
Combining this test with the initial purity tests gives, for each fixed generator, a QPT distinguisher using polynomially many copies with constant distinguishing advantage.
Together with the existence of \OWF established above, this gives an oracle separation: relative to $\mathcal O=(\mathcal R,\mathsf{QPSPACE})$, quantum-secure \OWF exist, but \PRS with superlogarithmic output length do not.


\section{Preliminaries}

\paragraph{Notation.}
We denote the security parameter by $\lambda \in \mathbb{N}$.
We use $[N]$ for $\{1, \ldots, N\}$.
A function $f : \mathbb{N} \to \mathbb{R}_{\geq 0}$ is \emph{negligible} if, for every positive polynomial $p$, there exists $\lambda_0 \in \mathbb{N}$ such that $f(\lambda) \leq 1/p(\lambda)$ for every $\lambda \geq \lambda_0$.
We use $\mathrm{negl}(\lambda)$ to denote a negligible function of $\lambda$.
For a finite set $S$, $x \leftarrow S$ means $x$ is sampled uniformly at random from $S$.
More generally, $x \leftarrow \mathcal{D}$ means $x$ is sampled according to the distribution $\mathcal{D}$.
We write $p$ as shorthand for any function $p(\lambda)$ in the security parameter $\lambda$ when no confusion can arise.

\subsection{Cryptographic primitives}\label{subsec:crypto}

We first recall the definitions of the classical and quantum cryptographic primitives considered in this work.
Throughout, all algorithms and adversaries are uniform, and computational security is defined against quantum polynomial-time (QPT) adversaries.
All length functions below are efficiently computable and polynomially bounded.

Our first main primitive is one-way functions, whose existence is necessary for essentially all of classical cryptography~\cite{IL89}.
\begin{definition}[One-way functions]
    Let $\lambda\in\mathbb N$ and $n(\lambda)$ be the input and output lengths, respectively.
    A function family $f = \{ f_\lambda : \{0,1\}^{\lambda} \to \{0,1\}^{n} \}_{\lambda \in \mathbb{N}}$ is a \textbf{(quantum-secure) one-way function family} if the following two conditions hold:
    \begin{enumerate}
        \item (Efficient computation). There is a deterministic polynomial-time classical algorithm that computes $f_\lambda(x)$ on input $(1^\lambda,x)$, where $x\in\{0,1\}^{\lambda}$.
        \item (Hard to invert). For every QPT adversary $\mathcal{A}$,
        \begin{align} 
            \Pr_{ x \leftarrow \{0,1\}^{\lambda}} \left[ f_\lambda(x) = f_\lambda \big(\mathcal{A}(1^\lambda,f_\lambda(x)) \big) \right] \le \mathrm{negl}(\lambda). 
        \end{align}
    \end{enumerate}
\end{definition}

We also consider the following classical cryptographic primitives.
A pseudorandom generator expands a uniformly random seed into a longer string that is computationally indistinguishable from a uniformly random string.
\begin{definition}[Pseudorandom generators]\label{def:prg}
    Let $\lambda \in \mathbb{N}$ and $\ell(\lambda)>\lambda$ be the seed and output lengths, respectively.
    A function family $G=\{G_\lambda:\{0,1\}^{\lambda}\to\{0,1\}^{\ell}\}_{\lambda\in\mathbb N}$ is a \textbf{(quantum-secure) pseudorandom generator} if the following two conditions hold:
    \begin{enumerate}
        \item (Efficient computation). There is a deterministic polynomial-time classical algorithm that computes $G_\lambda(s)$ on input $(1^\lambda,s)$, where $s\in\{0,1\}^{\lambda}$.
        \item (Pseudorandomness). For every QPT adversary $\mathcal A$,
        \begin{align}
            \left|\Pr_{s\leftarrow\{0,1\}^{\lambda}}
            [\mathcal A(1^\lambda,G_\lambda(s))=1]
            -\Pr_{u\leftarrow\{0,1\}^{\ell}}
            [\mathcal A(1^\lambda,u)=1]\right|
            \leq\mathrm{negl}(\lambda).
        \end{align}
    \end{enumerate}
\end{definition}

A pseudorandom function family consists of keyed functions that are computationally indistinguishable from a uniformly random function when the key is chosen uniformly at random.
Here, we define \textsf{PRFs} with classical queries, allowing the adversary to perform quantum computation but restricting its queries to the challenge function to be classical~\cite{Zha12,Zha25}.
\begin{definition}[Pseudorandom functions with classical queries]\label{def:prf}
    Let $d(\lambda)$ and $\ell(\lambda)$ be the input and output lengths, respectively.
    A keyed function family $F=\{F_k:\{0,1\}^{d}\to\{0,1\}^{\ell}\}_{k\in\{0,1\}^{\lambda}}$ is a \textbf{(quantum-secure) pseudorandom function family with classical queries} if the following two conditions hold:
    \begin{enumerate}
        \item (Efficient computation). There is a deterministic polynomial-time classical algorithm that computes $F_k(x)$ on input $(1^\lambda,k,x)$, where $k\in\{0,1\}^{\lambda}$ and $x\in\{0,1\}^{d(\lambda)}$.
        \item (Pseudorandomness). For every QPT adversary $\mathcal A$ making polynomially many adaptive classical queries to its challenge function,
        \begin{align}
            \left|\Pr_{k\leftarrow\{0,1\}^{\lambda}}
            [\mathcal A^{F_k}(1^\lambda)=1]
            -\Pr_{h\leftarrow  \mathrm{Func}(d,\ell)}
            [\mathcal A^{h}(1^\lambda)=1]\right|
            \leq\mathrm{negl}(\lambda),
        \end{align}
        where $\mathrm{Func}(d,\ell)$ is the set of all functions from $\{0,1\}^{d}$ to $\{0,1\}^{\ell}$.
    \end{enumerate}
\end{definition}

We also use secret-key encryption~(\textsf{SKE}) and message authentication codes~(\textsf{MAC}) with their standard correctness requirements and, respectively, IND-CPA and EUF-CMA security against QPT adversaries making classical queries.
For interactive bit commitments, we require computational hiding against QPT receivers and statistical binding against arbitrary quantum senders, with all messages and openings classical.
For the classical scheme used here, binding means that the commitment transcript admits valid openings to both bits only with negligible probability.
We refer to~\cite{Gol01,Gol04,Nao91} for formal definitions.

Our second main primitive is pseudorandom quantum states, whose existence is a potentially weaker assumption for quantum cryptography than the existence of \OWF.
In this work, we consider the more general variant in which the generator output can be a mixed state.
\begin{definition}[Pseudorandom quantum states]
    Let $\lambda \in \mathbb{N}$ be the security parameter, and let $m(\lambda)$ be the number of qubits in the quantum system. A keyed family of quantum states $\{\rho_k \}_{k \in  \{0, 1\}^\lambda}$ is \textbf{pseudorandom} if the following two conditions hold:
    \begin{enumerate}
        \item (Efficient generation). There is a QPT algorithm $G$ that generates the $m$-qubit state $\rho_k$ on input $(1^\lambda,k)$. 
        \item (Pseudorandomness). Any polynomially many copies of $\rho_k$ with the same random $k \in \{0, 1\}^\lambda$ are computationally indistinguishable from the same number of copies of a Haar-random state. More precisely, for any QPT adversary $\mathcal{A}$ and any $r\in \mathrm{poly}(\lambda)$, 
        \begin{align}
            \left| \Pr_{k \leftarrow \{0, 1\}^\lambda } [\mathcal{A}(1^\lambda,G(1^\lambda, k)^{\otimes r} ) = 1] - \Pr_{\ket{\psi} \leftarrow \mu_m} [\mathcal{A}(1^\lambda, \ket{\psi}^{\otimes r}) = 1] \right| 
            \le \mathrm{negl}(\lambda),
        \end{align}
        where $\mu_m$ is the Haar measure on $m$-qubit states.
    \end{enumerate}
\end{definition}


\subsection{Quantum property tests}

We will use the following two well-known results for breaking \PRS.

\begin{lemma}[SWAP test and purity test]
    For all states $\rho$ and $\sigma$, the SWAP test outputs $1$ with probability
    $(1 + \Tr(\rho \sigma))/2$. In particular, if $\sigma = \rho$, it is called a purity test and outputs $1$
    with probability $(1 + \Tr(\rho^2))/2$.
\end{lemma}

\begin{lemma}[Quantum OR lemma \cite{Harrow_Sequential_2017}]
    Let $\Pi_1, \ldots, \Pi_N$ be a sequence of projectors, where $\Pi_i$ corresponds to the
    two-outcome measurement $M_i = \{\Pi_i, I - \Pi_i\}$, and let $0 \le \epsilon \le 1/2$ and
    $\delta > 0$. There is a quantum OR test that uses one copy of a state $\rho$ and satisfies
    the following.
    \begin{itemize}
        \item (Case 1) If there exists $i \in [N]$ such that $\Tr(\Pi_i \rho) \ge 1 - \epsilon$,
        then it accepts with probability at least $(1 - \epsilon)^2/7$.
        \item (Case 2) If $\frac{1}{N} \sum_{i \in [N]}\Tr(\Pi_i \rho) \le \delta$, then it accepts with probability at most $4N\delta$.
    \end{itemize}
\end{lemma}
The OR test allows us to combine many candidate measurements without requiring a separate copy of the input state for each candidate.

    
    

\section{Oracle Model}\label{sec:oracle}
\subsection{Oracles for quantum algorithms}

We first review the oracle interfaces relevant to our separation.

\paragraph{Quantum oracles.}
Quantum oracles can provide different forms of access to the same underlying state or operation.
For quantum operations, common oracle interfaces provide access to a unitary and its inverse, a unitary alone, an isometry, or a general quantum channel.
Every quantum channel $\mathcal E$ admits a Stinespring isometric dilation $V$ and a unitary extension $U$ with an environment register $E$, such that
\begin{align}
    \mathcal{E}(\rho)
    = \Tr_E\!\left[V\rho V^\dagger\right]
    = \Tr_E\!\left[
        U\left(\rho\otimes\ket{0}\bra{0}_E\right)U^\dagger
    \right].
\end{align}
In the channel interface, the environment $E$ is traced out, and only the reduced output state is returned to the caller.
Thus, although quantum channels describe more general quantum operations, oracle access to a channel can be more limited than access to its unitary extension.

\paragraph{Quantum oracles for classical functions.}
Even when the underlying object is a classical function $f$, different oracle implementations can provide different forms of access to $f$.
Let $X$ and $Y$ denote the query and response registers, respectively.
Coherent access to $f$ is commonly implemented by a unitary acting as
\begin{align}
    U_f(\ket{x}_X\ket{y}_Y) = \ket{x}_X \ket{y \oplus f(x)}_Y.
\end{align}
This is the standard form of coherent access to a classical function, which we call a \textit{quantumly-accessible} oracle.
In the \textit{classical-accessible} oracle used here, the oracle first measures $X$ in the computational basis and then XORs $f(x)$ into $Y$ for the measured input $x$. 
Thus, the oracle destroys coherence between distinct query inputs.

\subsection{Our oracle}
We now give formal definitions of the two components of the oracle $\mathcal{O} = (\mathcal{R}, \mathsf{QPSPACE})$.

\paragraph{Classical-accessible random oracle.}
We use a classical-accessible random oracle $\mathcal R$~\cite{BDF+11,AF22}, defined as follows.
For every security parameter $\lambda\in\mathbb N$, independently sample a function $R_\lambda:\{0,1\}^{\lambda}\to\{0,1\}^{n(\lambda)}$ uniformly at random, where $n(\lambda)\ge\lambda$.
Each function $R_\lambda$ is sampled once and remains fixed for all queries. 
We denote the sampled function family by $R=\{R_\lambda\}_{\lambda\in\mathbb N}$ and its classical-accessible channel interface by $\mathcal R$.
The oracle acts on an arbitrary state $\rho_{XY}$ as 
\begin{align}
    \mathcal{R}(\rho_{XY})
    =
    \sum_{x\in\{0,1\}^{\lambda}}
    K_x \rho_{XY} K_x^\dagger,
    \quad
    K_x
    =
    \ket{x}\!\bra{x}_X
    \otimes
    \sum_{y\in\{0,1\}^{n}}
    \ket{y\oplus R_\lambda(x)}\!\bra{y}_Y.
\end{align}
This is a well-defined CPTP map implementing the classical-accessible evaluation of $R_\lambda$ described in the previous subsection.
When $Y$ is initialized to $\ket{0^{n}}$, its state conditioned on the measurement outcome $x$ is $\ket{R_\lambda(x)}$.

\paragraph{$\mathsf{QPSPACE}$ oracle.}
Following~\cite{Chen_Power_2025,BMMMY25}, we use a \textsf{QPSPACE} oracle to execute polynomial-width unitary computations on quantum inputs.
Roughly, it takes a classical description of an oracle-free Turing machine $M$, a time bound $t$ encoded in binary, a quantum data register, and an abort qubit.\footnote{Note that, in this work, a \textsf{QPSPACE} oracle does not refer to an oracle that takes a classical input $x$ and decides whether $x \in L$ using polynomial quantum space with bounded error (i.e., \textsf{BQPSPACE}). Rather, it takes a quantum state as input, applies a polynomial-width quantum circuit with potentially exponentially many gates to that state, and returns the resulting quantum state.}
We consider an interface in which the command registers encoding the Turing machine description and the time bound are measured in the computational basis at the start of each call, thereby destroying coherence between distinct classical commands.
The oracle runs $M$ for at most $t$ steps.
If $M$ halts within $t$ steps and outputs a valid oracle-free unitary quantum circuit on exactly the data register, then it applies that circuit to the data register.
Note that $t$ is encoded in binary; hence, a polynomial-length command can specify exponentially many gates for the \textsf{QPSPACE} oracle to execute.

More formally, the oracle $\mathcal{O}$ is defined as follows:

\begin{definition}[Oracle model]
    The oracle $\mathcal{O} = (\mathcal{R}, \mathsf{QPSPACE})$ is defined as follows:
    \begin{itemize}
        \item Classical-accessible random oracle $\mathcal{R}$,
        \begin{itemize}
            \item Input: $\rho_{XY}$, where $X$ is a $\lambda$-qubit query register and $Y$ is an $n$-qubit response register.
            \item Output: $\mathcal{R}(\rho_{XY})$.
        \end{itemize}
        \item $\mathsf{QPSPACE}$ oracle,
        \begin{itemize}
            \item Input: command registers encoding $(\langle M\rangle,\operatorname{bin}(t))$, a quantum data register, and an abort qubit initialized to $\ket{0}$.
            \item Action: measure the command registers in the computational basis and run the specified oracle-free Turing machine $M$.
            If $M$ halts in at most $t$ steps and outputs a valid oracle-free unitary circuit $C$ acting on exactly the data register, apply $C$ to that register.
            Otherwise, leave the data register unchanged and flip the abort qubit.
            \item Output: the measured command registers, the data register, and the abort qubit.
    \end{itemize}
    \end{itemize}
\end{definition}

\section{Constructions of \OWF}
The following theorem establishes the existence of quantum-secure \OWF in our oracle world. The proof is standard, but we include it for completeness.
\begin{theorem}\label{thm:main-1}
    With probability $1$ over the choice of the random functions ${R} = \{R_\lambda\}_{\lambda \in \mathbb{N}}$, the induced classical function family $f = \{ f_\lambda\}_{\lambda\in\mathbb{N}}$ is a quantum-secure \OWF relative to $\mathcal{O}$.
\end{theorem}
\begin{proof}
    We simply define the classical function $f_\lambda(x) \coloneqq R_\lambda(x)$, for all $x \in \{0, 1\}^\lambda$. 
    The function family is efficiently computable relative to $\mathcal{O}$ since the computation can be implemented by a single query to the oracle $\mathcal{R}$ as follows:
    first prepare the state $\ket{x}_X \ket{0^n}_Y$, then query $\mathcal{R}$, and finally measure the response register $Y$ in the computational basis.\footnote{This construction uses the oracle only on computational-basis states, and hence constitutes classical oracle access.}
    
    It remains to prove that $f$ is an \OWF against any quantum polynomial-time adversary with access to the oracle $\mathcal{O}$. 
    Note that since the oracle $\mathsf{QPSPACE}$ does not depend on the random oracle  $\mathcal{R}$, we may regard the adversary together with all of its $\mathsf{QPSPACE}$-queries as a single adaptive strategy; only its $R_\lambda$-queries access fresh information about $R_\lambda$. Also, at security parameter $\lambda$, the adversary might query other random-oracle components $R_{\lambda'}$ with $\lambda' \ne \lambda$.
    However, the random oracles at distinct security parameters are independent. Thus, we condition on all such independent components and analyze only the queries to $R_\lambda$.
    
    Suppose the challenger samples $x \leftarrow \{0, 1\}^\lambda$, and gives $y \coloneqq f_{\lambda}(x)$ to the adversary.
    Since $\mathcal A$ is a QPT adversary, it makes at most $q(\lambda)-1$ queries to $R_\lambda$ for some polynomial $q$.
    We count the adversary's final output $x'=\mathcal A^\mathcal O(1^\lambda,y)$ as an additional hypothetical query, since $x'$ may satisfy $R_\lambda(x')=y$ even if it was never queried. \footnote{The adversary may also query $R_{\lambda'}$ with $\lambda' \ne \lambda$ and $\mathsf{QPSPACE}$; throughout this proof, \textit{query} means a query to $R_\lambda$.}
    Because $\mathcal{R}$ first measures $X$, each query to the channel oracle can be simulated exactly by one classical query $x_i \mapsto R_\lambda(x_i)$ followed by the corresponding XOR shift on $Y$; hence the pair $(x_i, R_\lambda(x_i))$ is well defined even when the query and response registers are entangled with the adversary's workspace.
    
    Let $(x_i, y_i), i \in [q]$ denote the $i$th queried inputs and responses, padding the list if needed and treating the final output $x'$ as a hypothetical last query, where $y_i = R_\lambda(x_i)$.
    Let $E_i$ be the event that the $i$th query finds a preimage of $y$, i.e., $\{y_i = y\}$.
    For a fixed choice of the function family $R$, denote the success probability of the adversary $\mathcal{A}$ by
    \begin{align}
        p_{\mathcal{A}, \lambda}(R) \coloneqq \Pr_{x \leftarrow \{0, 1\}^\lambda} \left[f_\lambda(x) = f_{\lambda}\big(\mathcal{A}^{\mathcal{O}}(1^\lambda, y)\big)\right].
    \end{align}
    We now bound the expectation of this success probability over the choice of the random function family $R$. This is bounded as
    \begin{align}
         \mathbb{E}_R \left[ p_{\mathcal{A}, \lambda}(R)  \right] \le \Pr[\bigcup_{i=1}^q E_{i}].
    \end{align}
    Define the disjoint events $F_{i}$, $i \in [q]$, to denote that the first inversion event occurs on the $i$th query, i.e., $F_i = \bar{E}_{1} \cap \cdots \cap \bar{E}_{i-1} \cap {E}_{i}$.
    Then,
    \begin{align}
        \mathbb{E}_R\left[ p_{\mathcal{A}, \lambda}(R)\right] &\le  \Pr[\bigcup_{i=1}^q F_{i}] \\
        &= \sum_{i=1}^q \Pr[\bar{E}_{1} \cap \cdots \cap \bar{E}_{i-1} \cap {E}_{i}] \\
        &\le  \sum_{i=1}^q \Pr[{E}_{i} \mid  \bar{E}_{1} \cap \cdots \cap \bar{E}_{i-1} ],
    \end{align}
    where the first line uses $\bigcup_i E_i=\bigcup_i F_i$, the equality uses disjointness, and the last inequality follows from the chain rule.
    For simplicity, we denote the failure history $\bar{E}_{1} \cap \cdots \cap \bar{E}_{i-1}$ by $H_i$.
    
    There are two cases for a successful query.
    First, the adversary queries the planted input $x$, that is, $x_i=x$. Second, it queries another preimage $x_i\in f_\lambda^{-1}(y)$ with $x_i\ne x$.
    
    Conditioned on a complete failure transcript, the planted input is uniform among the unqueried inputs. There are at least $2^\lambda-(i-1)$ such inputs. At any other fresh input, the function value remains uniform in $\{0,1\}^n$; a repeated query cannot succeed after this failure history. Averaging over these transcripts, the success probability of the $i$th attempt is bounded as
    \begin{align}
        \Pr[E_i |H_i] &=  
        \Pr[\{x_i = x\} \mid  H_i] \cdot \Pr[E_i\mid \{x_i = x\} \cap H_i]\\
        &\qquad +
        \Pr[\{x_i \ne x\} \mid H_i]  \cdot  \Pr[E_i \mid \{x_i \ne x\}\cap H_i] \\
        &\le \frac{1}{2^{\lambda}-{(i-1)}} + \frac{1}{2^n}.  
    \end{align}
    Therefore, we obtain
    \begin{align}
       \mathbb{E}_R\left[ p_{\mathcal{A}, \lambda}(R)\right] &\le  \sum_{i=1}^q \left( \frac{1}{2^{\lambda} - (i-1)} + \frac{1}{2^n}\right) \\
        &\le \frac{q}{2^\lambda -q +1} + \frac{q}{2^n} \le \frac{3q}{2^{\lambda}},
    \end{align}
    where the last inequality holds for all sufficiently large $\lambda$, since $q$ is polynomial and $n\ge\lambda$.
    By Markov's inequality,
    \begin{align}
        \Pr_R[  p_{\mathcal{A}, \lambda}(R) \ge 2^{- \lambda / 4} ] \le 2^{\lambda/4} \left(\frac{3q}{2^\lambda} \right).
    \end{align}
    Since $q$ is polynomial in $\lambda$, we have $\sum^{\infty}_{\lambda=1} 3q(\lambda)2^{-3\lambda/4}<\infty$.
    By the Borel-Cantelli lemma, with probability $1$ over $R$, 
    \begin{align}
        p_{\mathcal{A}, \lambda}(R) \le 2^{- \lambda/4},
    \end{align}
    for all sufficiently large $\lambda$.
    Since there are only countably many QPT adversaries, this bound holds simultaneously for every QPT adversary.
    Therefore, with probability $1$ over $R$, the function family $f$ is classically computable and quantum-secure one-way relative to $\mathcal{O}$.

    Note that we do not assume that $\mathcal{A}$ is classical. $\mathcal{A}$ may use arbitrary quantum computation and entangle its workspace with the query and response registers. However, since the oracle $\mathcal{R}$ first measures the input register, each oracle call is exactly reducible to one adaptively chosen classical input-output pair, as formalized above.
\end{proof}


Theorem~\ref{thm:main-1} also gives the following classical primitives.
\begin{corollary}\label{cor:classical}
    With probability $1$ over the choice of the random functions ${R} = \{R_\lambda\}_{\lambda \in \mathbb{N}}$, the following primitives exist relative to $\mathcal O=(\mathcal R,\mathsf{QPSPACE})$:
    \begin{itemize}
        \item quantum-secure $\mathsf{PRGs}$ with any polynomial output length,
        \item quantum-secure $\mathsf{PRFs}$ with classical queries,
        \item IND-CPA-secure $\mathsf{SKE}$ and EUF-CMA-secure $\mathsf{MAC}$,
        \item statistically binding, computationally hiding interactive bit commitments.
    \end{itemize}
    All honest algorithms are classical polynomial-time algorithms using only classical queries to $\mathcal R$.
    Computational security holds against QPT adversaries with the full access to $\mathcal O$ prescribed by our oracle model.
\end{corollary}
\begin{proof}
    Fix $R$ satisfying Theorem~\ref{thm:main-1}.
    The HILL and GGM constructions give quantum-secure \textsf{PRGs} and \textsf{PRFs} with classical queries~\cite{HILL99,GGM86,Zha12,Zha25}.
    Their reductions evaluate the underlying functions only on classical inputs and run the adversary with its prescribed $\mathcal O$-access unchanged.
    Consequently, a non-negligible attack gives a QPT inverter for $f_\lambda=R_\lambda$ relative to the same $\mathcal O$.
    Standard \textsf{PRF}-based constructions give the encryption and authentication schemes~\cite{GGM85}, and Naor's construction gives the commitments~\cite{Nao91}.
    Their computational-security reductions have the same access property.
    For binding, use a \textsf{PRG} with seed length $\lambda$ and output length $3\lambda$.
    For every fixed $R$, the receiver's uniform $3\lambda$-bit challenge belongs to the set of XOR differences of two \textsf{PRG} outputs with probability at most $2^{2\lambda}/2^{3\lambda}=2^{-\lambda}$.
    Outside this event, no commitment transcript has valid openings to both bits, regardless of the sender's computation or oracle access.
    All honest algorithms use only classical evaluations of $f$, so the conclusion holds on the probability-one event of Theorem~\ref{thm:main-1}.
\end{proof}


Together, Corollary~\ref{cor:classical} and Theorem~\ref{thm:main-2} show that these classical primitives can exist in an oracle world without \PRS of output length $\omega(\log\lambda)$.

\section{Oracle Separation of \PRS from \OWF}

Finally, we prove our main result.
\begin{theorem}[]\label{thm:main-2}
    For every choice of the random functions $R$, there is no \PRS with output length $m(\lambda) = \omega(\log \lambda)$ relative to the oracle $\mathcal{O}$.
\end{theorem}

In this section, we denote the density matrix of a pure state $\ket{\phi}$ as $\rho_{\phi} \coloneqq \ket{\phi}\bra{\phi}$, and for a fixed $R$, we write $\rho_k \coloneqq \Gen^{\mathcal O}(1^\lambda,k)$. 
In the following subsections, we first analyze the properties of an arbitrary QPT generator that may produce mixed outputs. 
We show that its output state either has low purity or has a large overlap with a particular pure state in its mixed-state decomposition that occurs with at least inverse-exponential probability. 
This dichotomy suggests the following attack. In the low-purity case, a purity test suffices to detect the output. 
Otherwise, we construct preparation circuits for all possible branches and use an OR test to detect whether any candidate has sufficiently large overlap with the output state.
Finally, we show that the OR test can be implemented through $\mathsf{QPSPACE}$ and analyze the distinguishability advantage of the adversary.

\subsection{Constructing an oracle-free branch}

\paragraph{Dichotomy.}
Since we allow the \PRS generator to perform intermediate measurements, query the quantum channel oracle $\mathcal{O}$, and discard work registers, the resulting \PRS state may naturally be mixed. 
The following lemma establishes a useful property of general mixed states. It states that if a mixed state is nearly pure, then a branch of probability at least $\frac{1}{2N}$ has a large overlap with the mixed state.

\begin{lemma}\label{lem:dichotomy}
    Let $\rho=\sum_{i=1}^{N}p_i\ket{\psi_i}\bra{\psi_i}$ be a state with a finite pure-state decomposition, and let $0<\delta<1/2$. Then, one of the following holds:
    \begin{itemize}
        \item (Case 1) $1- \Tr(\rho^2) > \delta$, or
        \item (Case 2) $1- \Tr(\rho^2) \le \delta$ and there is an index $i$ such that 
        \begin{align}
             p_i \ge \frac{1}{2N}, \quad \langle \psi_i | \rho | \psi_i \rangle \ge 1 - 2 \delta. 
        \end{align}
    \end{itemize}
\end{lemma}
\begin{proof}
    Suppose that $1 - \Tr(\rho^2) \le \delta$. From the assumption, we have
    \begin{align}
        \sum_{i=1}^N p_i (1 - \langle{\psi_i | \rho | \psi_i}\rangle) = 1 -{\Tr}(\rho^2)\le \delta.
    \end{align}
    By Markov's inequality, the indices with overlap at least $1-2 \delta$ have total probability at least $1/2$. Since there are at most $N$ such indices, one index has probability at least $\frac{1}{2N}$.
\end{proof}

In the first case, repeated purity tests suffice to detect mixedness.
In the second case, the lemma guarantees only the existence of such a branch. The adversary does not know which branch satisfies the guarantee, and there are exponentially many possible branches. 
Moreover, to preserve the security of \OWF, the oracle $\mathsf{QPSPACE}$ cannot directly execute $\Gen^{\mathcal{O}}$. We therefore introduce the following lemma.

\paragraph{Oracle-free branch.} 
We apply the above lemma to the pure-state decomposition induced by the complete execution branches of $\Gen^{\mathcal O}$. To define these branches, we retain discarded registers and include the outcomes of a final measurement of all nonoutput registers. Distinct complete records together with their input keys receive distinct descriptors $d$ of polynomial length because $\Gen^{\mathcal O}$ runs in polynomial time.

\begin{lemma}\label{lem:oracle_free}
For any uniform QPT generator $\Gen^{\mathcal O}$, there exist a polynomial $L=L(\lambda)$ and a family of oracle-free unitaries $\{U_d\}_{d\in\{0,1\}^L}$ on $L$ qubits satisfying the following:
\begin{enumerate}
    \item \textup{(Uniform construction.)}
    A single $R$-independent Turing machine, given $(1^\lambda,d)$, prints $U_d$ in time $2^{O(L)}$.
    \item \textup{(Branch reproduction.)}
    For every fixed $R$ and key $k$, each actual branch of $\Gen^{\mathcal O}(1^\lambda,k)$ has a distinct descriptor $d\in\{0,1\}^L$.
    Postselecting the nonoutput qubits of $U_d|0^L\rangle$ on the all-zero outcome succeeds with exactly that branch's probability and prepares its normalized output state.
\end{enumerate}
\end{lemma}

\begin{proof}
    Fix $R$ and $k$. Our goal is to construct, for each complete branch of $\Gen^{\mathcal O}(1^\lambda,k)$, an oracle-free unitary that reproduces the branch's probability and output state upon postselection. 
    We first handle the operations of $\Gen^{\mathcal O}$ other than its oracle calls.
    By unitary dilation and the principle of deferred measurement, we can implement these operations unitarily, coherently recording intermediate measurement outcomes and retaining the environment registers. 
    For a proposed record, fix subsequent classical controls to the prescribed outcomes. 
    Together with a final measurement of all nonoutput registers, the recorded outcomes specify complete branches $i$ and give
    \begin{align}
        \rho_{k}=\sum_i p_i|\psi_i\rangle\langle\psi_i|,
    \end{align}
    where each record $i$ also includes the inputs measured by $\mathcal R$ and the $\langle M \rangle, \mathrm{bin}(t)$ measured by $\mathsf{QPSPACE}$.
    
    We next remove the oracle calls. 
    Once a measured command is fixed, $\mathsf{QPSPACE}$ runs the oracle-free circuit specified by that command (or applies the prescribed abort flip), so we can insert that operation directly. At an $\mathcal R$ call, the measured input is already part of the branch record, but its reply is unknown. We therefore allow every possible reply as a candidate.
    
    Let $d$ contain $k$, the proposed complete record, and the proposed replies. Use these values to construct an oracle-free unitary, while recording the inputs and commands measured by the oracles for final postselection. If $d$ matches an actual branch $i$, the fixed oracle actions agree with those of that branch. Unitary dilation and deferred measurement preserve its unnormalized output and probability, so postselection gives $\sqrt{p_i}|\psi_i\rangle$. Flipping the nonoutput bits specified by $d$ yields a unitary $U_d$ for which the same postselection outcome is all zero for every candidate. Thus it succeeds with probability $p_i$ and prepares $|\psi_i\rangle$.
    
    Finally, $\Gen^{\mathcal O}$ uses polynomially many qubits and oracle calls. The polynomial-length binary time bound of its $\mathsf{QPSPACE}$ calls allows their circuits to be printed in $2^{\mathrm{poly}(\lambda)}$ time. Choose $L\ge \max\{\lambda, m(\lambda)\}$ large enough for the descriptors, circuit width, and all command lengths, and extend every $U_d$ to $L$ qubits by adding all-zero qubits. Malformed descriptors are assigned the identity circuit. A single uniform Turing machine, independent of $R$, then prints $U_d$ from $(1^\lambda,d)$ in time $2^{O(L)}$. Distinct actual branches have distinct length-$L$ descriptors, so there are at most $2^L$ of them.
\end{proof}

\subsection{Amplified candidate states for OR test}
By Lemma~\ref{lem:oracle_free}, we can enumerate a family of oracle-free unitaries that contains, for every key $k$ and every random function family $R$, a unitary corresponding to each possible execution branch of the \PRS generator. 
We use projectors associated with these unitaries as candidates for an OR test that distinguishes \PRS outputs from Haar-random states.

Of the $L$ qubits on which $U_d$ acts, let $A$ denote the first $m$ output qubits and $E$ the remaining $L-m$ qubits. For each descriptor $d$, let $p_d$ be the probability of obtaining $|0^{L-m}\rangle_E$ after applying $U_d$ to $|0^L\rangle$. 
If $p_d > 0$, let $|\psi_d\rangle_A$ be the corresponding normalized pure state on $A$.
For an actual branch $i$ with descriptor $d$, we have $p_d=p_i$ and $|\psi_d\rangle=|\psi_i\rangle$ where $p_i$ and $\ket{\psi_i}$ are the corresponding branch probability and pure state in $\rho_k  = \sum_{i} p_i \ket{\psi_i} \bra{\psi_i}$. Given an $m$-qubit state $\rho$, define
\begin{align}
    \omega_\rho  \coloneqq
    (\rho\otimes|0^{L-m}\rangle\langle0^{L-m}|)^{\otimes L},
    \quad \Lambda_d \coloneqq
    (U_d|0^L\rangle\langle0^L|U_d^\dagger)^{\otimes L}.
\end{align}
To apply the completeness guarantee of the quantum OR lemma with $\epsilon=1/2$ when the challenge state is a \PRS, it suffices to find a descriptor $d$ for which ${\Tr}(\Lambda_d \omega_{\rho_{k}})\ge 1/2$. 
In Case 2 of Lemma~\ref{lem:dichotomy}, taking $N=2^L$ and $\delta=\frac{1}{8L}$ guarantees an actual branch $i$ with $p_i\ge2^{-L-1}$ and $\langle\psi_i|\rho_{k}|\psi_i\rangle\ge1-\frac{1}{4L}$.
For its corresponding descriptor $d$, we obtain
\begin{align}
    \Tr(\Lambda_d \omega_{\rho_{k}}) = \bigl(p_d \langle\psi_d|\rho_k|\psi_d \rangle\bigr)^L
    \ge\biggl(\frac{1}{2^{L+1}}\Bigl(1-\frac1{4L}\Bigr)\biggr)^L.
\end{align}
Since this lower bound is below $1/2$, it does not establish the completeness condition for the OR test. 
We first amplify the preparation probability of each candidate branch.

We use the fixed-point amplification method of~\cite{Grover_fixed_2005}.
The recursion does not require the exact initial success probability: 
the lower bound from Lemma~\ref{lem:dichotomy} suffices to choose the number of steps. This particular $\pi/3$ construction does not achieve the $O(\sqrt{1/p_d})$ query scaling of ordinary amplitude amplification.
That cost is acceptable here because the amplified circuit has exponential length but polynomial width, and we will include it in the oracle-free OR-test circuit executed by $\mathsf{QPSPACE}$.

For each descriptor $d$, set $U_{d,0}=U_d$ and take $|0^L\rangle$ as the starting state. 
The target subspace consists of the states for which $E$ is all zero. 
Let $P\coloneqq I_A\otimes|0^{L-m}\rangle\langle0^{L-m}|_E$ be its projector. 
The two phase rotations of~\cite{Grover_fixed_2005} become
\begin{align}
    R_s \coloneqq I-(1-e^{i\pi/3})|0^L\rangle\langle0^L|, \quad 
    R_t \coloneqq I-(1-e^{i\pi/3})P.
\end{align}
Using these rotations, define the unitary at the next step recursively by
\begin{align}
    U_{d,j+1}\coloneqq U_{d,j}R_sU_{d,j}^\dagger R_tU_{d,j}.
\end{align}
If $p_{d,j}$ is the probability that $U_{d,j}|0^L\rangle$ lies in the target subspace, the fixed-point identity gives
\begin{align}
    p_{d,j+1}=1-(1-p_{d,j})^3.
\end{align}
Thus, the failure probability is cubed at each step. After $L$ steps, let $\widetilde U_d\coloneqq U_{d,L}$. Its success probability is
\begin{align}
    \widetilde p_d=1-(1-p_d)^{3^L}.
\end{align}
For the descriptor $d$ corresponding to the branch guaranteed by
Lemma~\ref{lem:dichotomy}, the initial success probability is $p_d\ge2^{-L-1}$. 
Hence, for all sufficiently large $\lambda$,
\begin{align}
    1-\widetilde p_d = (1-p_d)^{3^L}
    \le e^{-3^L/2^{L+1}}
    \le\frac1{4L}.
\end{align}
When applied to $|0^L\rangle$, each update stays within the span of $U_{d,j}|0^L\rangle$ and $PU_{d,j}|0^L\rangle$. Its component in the target subspace therefore retains its direction. 
For the descriptor $d$ of branch $i$, postselection on $E$ still prepares $\ket{\psi_i}$ on $A$.

Each step uses a constant number of copies of the preceding circuit and its inverse. 
Thus, after $L$ steps, $\widetilde U_d$ has $3^L=2^{O(L)}$ copies of the original circuit, retains polynomial width, and can still be printed by an $R$-independent machine in time $2^{O(L)}$. 
From this point on, redefine $\Lambda_d$ using the amplified unitaries:
\begin{align}
    \Lambda_d\coloneqq
    (\widetilde U_d|0^L\rangle\langle0^L|
    \widetilde U_d^\dagger)^{\otimes L}.
\end{align}
For the descriptor $d$ guaranteed by Case~2, both $\widetilde p_d$ and $\langle\psi_d|\rho_k|\psi_d\rangle$ are at least $1-\frac{1}{4L}$. Therefore, in this case there exists a candidate satisfying
\begin{align}\label{eq:amplified-completeness}
    \Tr(\Lambda_d\omega_{\rho_{k}}) = \bigl(\widetilde p_d
    \langle\psi_d|\rho_k|\psi_d\rangle\bigr)^L
    \ge\left(1-\frac1{4L}\right)^{2L}
    \ge\frac9{16}>\frac12,
\end{align}
where the second inequality follows by applying Bernoulli's inequality to $(1-\frac{1}{4L})^L$.
Thus, the amplified projectors satisfy the completeness condition 
of the quantum OR lemma for the \PRS outputs in Case~2.

\subsection{The adversary and its analysis}
\paragraph{The adversary.}
The adversary can access the quantum channel oracle $\mathcal{O}$ and can perform intermediate measurements and discard registers. 
\begin{itemize}
    \item \textbf{Input}: $3L$ copies of a test state $\rho$.
    \item \textbf{Procedure}:
    \begin{enumerate}
        \item Perform $L$ purity tests on disjoint pairs of the first $2L$ copies. If any test returns the antisymmetric outcome (labeled 0), then return $b=1$.
        \item Otherwise, 
            \begin{enumerate}
                \item Prepare $\omega_\rho$ from the remaining $L$ copies, together with an $(L+3)$-qubit counter register initialized to $\ket{0}_C$, a flag register initialized to $\ket{0}_f$, and an index register initialized to $\ket{+}^{\otimes L}_J$ for the OR test, together with any required zero-initialized work qubits.
                \item Send the classical command $(\langle M\rangle,\mathrm{bin}(t))$, the prepared query registers, and an abort qubit initialized to $\ket{0}$ to $\mathsf{QPSPACE}$, where $M$ prints the unitary OR-test circuit within $t$ steps.
                \item Measure the counter register. If the outcome is nonzero, return $b=1$; otherwise, return $b=0$.
            \end{enumerate}
    \end{enumerate}
    \item \textbf{Output}: $b \in \{0, 1\}$, where $b=0$ denotes a Haar-random state and $b=1$ denotes a \PRS.
\end{itemize}
The above attack can distinguish \PRS outputs from Haar-random states in polynomial time with a single query to $\mathsf{QPSPACE}$.

\begin{lemma}
    The adversary $\mathcal{A}$ that distinguishes \PRS outputs from Haar-random states can be implemented by a QPT algorithm with a single query to $\mathsf{QPSPACE}$.
\end{lemma}
\begin{proof}
    We check that the unitary OR construction of~\cite[Remark~5.3 and Appendix~A]{Chen_Power_2025} applies in our setting and that all remaining operations of $\mathcal{A}$ have polynomial size.

    By Lemma~\ref{lem:oracle_free} and the amplification above, each measurement $\{\Lambda_d,I-\Lambda_d\}$ has a coherent implementation of polynomial width using $\widetilde U_d^{\otimes L}$, its inverse, and an all-zero test.
    These circuits can be printed uniformly from $d$ in time $2^{O(L)}$ by an $R$-independent procedure, which gives a polynomial-length description of the entire measurement family.
    To specify the accepting outcomes, let
    \begin{align}
        \Pi\coloneqq\sum_{d\in\{0,1\}^L}\Lambda_d\otimes|d\rangle\langle d|_J,
        \qquad
        \Delta\coloneqq I\otimes|+\rangle\langle+|_J^{\otimes L}.
    \end{align}
    The OR test alternates the measurements with accepting projectors $\Pi$ and $I-\Delta$. For $2^L$ candidates and $\epsilon=1/2$, the construction in~\cite{Harrow_Sequential_2017} uses $2^{L+1}$ rounds.
    For each accepting projector, we coherently compute the outcome into the flag, increment the counter when the flag is $1$, and uncompute the flag. An $(L+3)$-qubit counter prevents overflow throughout these $2^{L+2}$ tests.
    Since the counter only increases, its zero component is exactly the branch in which every test rejects, with operator $(\Delta(I-\Pi))^{2^{L+1}}$ on the input $\omega_\rho\otimes|+\rangle\langle+|^{\otimes L}_J$. Thus, accepting precisely when the final counter is nonzero reproduces the OR test's acceptance probability.
    The entire oracle-free circuit has polynomial width and can be printed by an $R$-independent classical Turing machine $M$ in time $2^{O(L)}$.

    For the fixed \PRS generator, the adversary can construct $\langle M\rangle$ from this printing procedure and prepare a sufficiently large binary time bound in polynomial time, and hence execute the entire circuit with a single query to $\mathsf{QPSPACE}$.
    All remaining operations, namely the purity tests, preparation of the query registers, and final counter measurement, have polynomial size.
\end{proof}
Now we are ready to prove our main theorem. The remaining task is to analyze the acceptance probability of each state. 

\begin{proof}[Proof of Theorem~\ref{thm:main-2}] 
    It remains to show that the QPT adversary described above has a non-negligible distinguishing advantage between the generator's output and a Haar-random state.

    The adversary's attack has two stages, a purity test and an OR test. Let $p_{\mathrm{purity}}(\rho)$ be the probability that at least one purity test returns the antisymmetric outcome $0$, and let $p_{\mathrm{OR}}(\rho)$ be the probability that the OR test returns $1$ on $\omega_\rho$. Since a single purity test returns the antisymmetric outcome with probability $(1-\Tr(\rho^2))/2$, we have
    \begin{align}
        p_{\mathrm{purity}}(\rho)
        =1-\left(\frac{1+\Tr(\rho^2)}{2}\right)^L.
    \end{align}
    The two stages use disjoint copies. The probability that the adversary accepts the given state, that is, decides that it is a \PRS output, is
    \begin{align}
        \Pr[\mathcal A^{\mathcal O}(1^\lambda,\rho^{\otimes3L})=1]
        =p_{\mathrm{purity}}(\rho)
        +\bigl(1-p_{\mathrm{purity}}(\rho)\bigr)p_{\mathrm{OR}}(\rho).
    \end{align}
    We analyze the adversary's acceptance probabilities on \PRS and Haar-random states by bounding the probability of each stage.

    \paragraph{Acceptance on \PRS.}
    First, fix $R$ and $k$. Apply Lemma~\ref{lem:dichotomy} with $N=2^L$ and $\delta=\frac{1}{8L}$. Depending on the purity of the mixed output state $\rho_k$ of $\Gen^\mathcal{O}$, there are two cases.

    If $1-\Tr(\rho_k^2)>\frac{1}{8L}$, then
    \begin{align}
        p_{\mathrm{purity}}(\rho_k)
        &=1-\left(\frac{1+\Tr(\rho_k^2)}{2}\right)^L \\[-2pt]
        &>1-\left(1-\frac{1}{16L}\right)^L
        \ge 1-e^{-1/16}>\frac{1}{17}.
    \end{align}
    In this case, the state's purity is already sufficiently low, so the purity test alone gives a sufficiently large acceptance probability.

    On the other hand, if $1-\Tr(\rho_k^2)\le\frac{1}{8L}$, we are in Case~2. Applying Eq.~\eqref{eq:amplified-completeness}, some candidate satisfies $\Tr(\Lambda_d\omega_{\rho_k})>1/2$. By the quantum OR lemma with $\epsilon=1/2$,
    \begin{align}
        p_{\mathrm{OR}}(\rho_k)\ge\frac{(1-\epsilon)^2}{7}=\frac1{28}.
    \end{align}
    The adversary's acceptance probability is at least $p_{\mathrm{OR}}(\rho_k)$, and hence at least $1/28$.

    Thus, for every fixed $R$ and $k$, the adversary accepts $\rho_k$ with probability at least $1/28$. Averaging over a uniformly random key $k$ gives
    \begin{align}
        \Pr_{k\leftarrow\{0,1\}^\lambda}
        [\mathcal A^{\mathcal O}(1^\lambda,\rho_k^{\otimes3L})=1]
        \ge\frac1{28}.
    \end{align}

    \paragraph{Acceptance on Haar-random states.}
    Let $\rho_\phi=|\phi\rangle\langle\phi|$. Since a Haar-random state is pure, ${\Tr}(\rho_\phi^2)=1$ and $p_{\mathrm{purity}}(\rho_\phi)=0$. It therefore suffices to bound the OR test's acceptance probability.

    We use the basic Haar-moment identity~\cite{Harrow_the_2013}. For any fixed pure state $|\psi\rangle$,
    \begin{align}
        \mathbb E_{\phi\leftarrow\mu_m}|\langle\psi|\phi\rangle|^{2L}
        =\binom{2^m+L-1}{L}^{-1}.
    \end{align}
    Thus, each candidate state $|\psi_d\rangle$ produced by each candidate circuit has a small average overlap with a Haar-random state. For each descriptor $d$ with $\widetilde p_d>0$,
    \begin{align}
        \mathbb E_{\phi\leftarrow\mu_m}{\Tr}(\Lambda_d\omega_{\rho_\phi})
        &=\mathbb E_{\phi\leftarrow\mu_m}
          \bigl(\widetilde p_d|\langle\psi_d|\phi\rangle|^2\bigr)^L \\[-2pt]
        &=\widetilde p_d^{\,L}\binom{2^m+L-1}{L}^{-1}
         \le\binom{2^m+L-1}{L}^{-1}.
    \end{align}
    The same upper bound holds if $\widetilde p_d=0$, since then ${\Tr}(\Lambda_d\omega_{\rho_\phi})=0$.

    Apply the quantum OR lemma to the averaged Haar input $\mathbb E_{\phi\leftarrow\mu_m}\omega_{\rho_\phi}$. By linearity of the acceptance probability, we obtain
    \begin{align}
        \Pr_{\phi\leftarrow\mu_m}
        [\mathcal A^{\mathcal O}(1^\lambda,\rho_\phi^{\otimes3L})=1]
        &=\mathbb E_{\phi\leftarrow\mu_m}p_{\mathrm{OR}}(\rho_\phi) \\[-2pt]
        &\le4\sum_{d\in\{0,1\}^L}
          \mathbb E_{\phi\leftarrow\mu_m}\Tr(\Lambda_d\omega_{\rho_\phi}) \\[-2pt]
        &\le4\cdot2^L\binom{2^m+L-1}{L}^{-1} \\
         &\le4\left(\frac{2L}{2^m}\right)^L 
         \le2^{2-2L}
         \le2^{-\lambda},
    \end{align}
    where the third inequality follows from $\binom{2^m+L-1}{L} \ge\left(\frac{2^m}{L}\right)^L$.
    Since $L$ is polynomial and $m=\omega(\log\lambda)$, we have $2^m\ge8L$ for all sufficiently large $\lambda$.

    Thus, for every possible \PRS generator $\Gen^{\mathcal O}$ and every $R$, the adversary distinguishes its output from a Haar-random state with non-negligible advantage, at least $1/28-2^{-\lambda}$. Hence \PRS with output length $\omega(\log \lambda)$ do not exist relative to $\mathcal O$.
\end{proof}

\section*{AI use disclosure}
OpenAI GPT-6 Astra was used to prepare an initial written draft of the proofs based on the authors' arguments. The authors subsequently revised and refined the proofs themselves, while GPT-6 Astra was also used to polish the written presentation. The authors take full responsibility for the content of this work.

\section*{Acknowledgments}
C.O. and V.S. were supported by the National Research Foundation of Korea Grants (No. RS-2024-00431768 and No. RS-2025-00515456) funded by the Korean government (Ministry of Science and ICT (MSIT)) and the Institute of Information \& Communications Technology Planning \& Evaluation (IITP) Grants funded by the Korean government (MSIT) (No. RS-2024-00437284, No. IITP-2025-RS-2025-02283189 and No. IITP-2025-RS-2025-02263264) by Global Partnership Program of Leading Universities in Quantum Science and Technology (RS-2025-08542968) through the National Research Foundation of Korea~(NRF) funded by the Korean government (Ministry of Science and ICT(MSIT)).

\bibliographystyle{alpha}
\bibliography{reference}

\end{document}

%% file: macro.tex
\ifnum\draft=1
\newcommand{\minki}[1]{\textcolor{blue}{$\langle\langle$Minki: #1$\rangle\rangle$}}
\else
\newcommand{\minki}[1]{}
\fi

\ifnum\draft=1
\newcommand{\vaughn}[1]{\textcolor{cyan}{$\langle\langle$Vaughn: #1$\rangle\rangle$}}
\else
\newcommand{\vaughn}[1]{}
\fi

\newcommand{\OWF}{\ensuremath{\mathsf{OWFs}}\xspace}

\newcommand{\Gen}{\ensuremath{G}\xspace}
\newcommand{\PRS}{\ensuremath{\mathsf{PRS}}\xspace}

\newtheorem{theorem}{Theorem}[section]
\newtheorem{lemma}[theorem]{Lemma}
\newtheorem{corollary}[theorem]{Corollary}

\newtheorem{definition}[theorem]{Definition}